\documentclass[twocolumn]{autart}
\usepackage[cmex10]{amsmath}
\usepackage{amsfonts,amssymb,mathtools}
\usepackage{graphicx,xcolor}
\usepackage{mathptmx}
\usepackage{enumitem}
\usepackage{subcaption}
\usepackage{ntheorem}
\graphicspath{{Figures/}{./}}
\DeclareMathOperator*{\argmin}{arg\,min}
\DeclareMathOperator*{\argmax}{arg\,max}

\theoremstyle{plain}

\newtheorem{lemma}{Lemma}
\newtheorem{theorem}{Theorem}
\newtheorem{assumption}{Assumption}
\newtheorem{remark}{Remark}
\newenvironment*{proof}{{\it Proof.}}{\hfill $\square$\par}

\begin{document}
\begin{frontmatter}
\title{{Output-Positive Adaptive Control of Parabolic PDE-ODE Cascades}\thanksref{footnoteinfo}}
\thanks[footnoteinfo]{The material in this paper was not presented at any conference.}
\author[XMU]{Yun Jiang}\ead{34520241151611@stu.xmu.edu.cn} and
\author[XMU]{Ji Wang}\ead{jiwang@xmu.edu.cn}
\address[XMU]{Department of Automation, Xiamen University, Xiamen 361005, China}
\begin{abstract}
In this paper, we propose a safe adaptive boundary control strategy for a class of parabolic partial differential equation–ordinary differential
equation (PDE–ODE) cascaded systems with parametric uncertainties in both the PDE and ODE subsystems. The proposed design is built
upon an adaptive Control Barrier Function (aCBF) framework that incorporates high-relative-degree CBFs together with a batch least-squares identification (BaLSI)-based adaptive control that guarantees exact parameter identification in finite time. The proposed controller ensures the positivity, i.e., the safety, of the plant output state that is the furthest state from
the control input, as well as the exponential regulation of the overall
plant state to zero. Numerical
simulations are provided to demonstrate the effectiveness of the proposed
approach.
\end{abstract}
\end{frontmatter}

\section{Introduction}
\subsection{Boundary control of parabolic PDEs}
Parabolic partial differential equations (PDEs) are widely employed to model
dynamics in fluid flow, heat transfer, and chemical processes, with
applications ranging from sea ice melting and freezing
\cite{koga2018control} to continuous steel casting \cite{petrus2012enthalpy}
and lithium-ion battery systems \cite{koga2021state,tang2017state}. These
applications naturally lead to significant control and estimation problems
for parabolic PDEs. The backstepping approach \cite{krstic2008boundary} has
been recognized as a powerful method for boundary
stabilization/estimation of PDEs. The backstepping boundary control for
parabolic PDEs was proposed in \cite{liu2000backstepping} and
\cite{liu2003boundary}. Since then, numerous advancements in boundary control
and estimation of parabolic PDEs have been made over the past two decades,
including \cite{baccoli2015boundary,deutscher2015backstepping,orlov2017output,pisano2012boundary,vazquez2008control,vazquez2008analysis}. Besides the
aforementioned studies on parabolic PDEs, there has been significant interest
in parabolic PDE-ODE coupled systems, which arise in various physical
contexts, including coupled electromagnetic, mechanical, and chemical
reactions. The backstepping stabilization of parabolic PDEs coupled with linear
ODEs was initially explored in \cite{krstic2009compensating} for Dirichlet-type
boundary interconnections, while results for Neumann boundary interconnections
were later detailed in \cite{susto2010control} and \cite{tang2011state}. More
recently, this framework has been extended to stabilize complex chains where
general linear ODEs are interconnected with parabolic PDEs featuring diffusion
and counter-convection through distributed coupling, as achieved by the
backstepping procedure in \cite{xu2023stabilization}.

\subsection{Safe adaptive control}
The PDE control designs listed above mainly focus on stabilization rather than
safety. Safety is critical in applications such as autonomous
vehicles, robotics, and UAVs. For ODEs, safety is usually enforced by control
barrier functions (CBFs) through a quadratic program (QP) safety filter
\cite{ames2016control}. High relative degree CBFs were considered in \cite{nguyen2016exponential}, whose root is the non-overshooting control
design introduced in \cite{krstic2006nonovershooting}. For PDEs, the safe backstepping boundary control are addressed in \cite{wang2025safe}, and \cite{koga2023safe}, for hyperbolic and parabolic PDEs, respectively. 

Model uncertainties can invalidate safety guarantees derived
from fully known models, motivating the recent interest in
safe adaptive control. Most existing approaches are based
on adaptive Control Lyapunov Functions (aCLFs) \cite{krstic1995control} with
adaptive Control Barrier Functions (aCBFs) introduced in
\cite{taylor2020adaptive} to enforce safety under parametric uncertainties,   where the state is constrained to a subset of the original
safe set, i.e., it restricts the initial conditions to a subset of the safe set
and shrinks the safe set in inverse proportion to the adaptation gain. There are some
extensions \cite{lopez2020robust,maghenem2021adaptive} developed to alleviate 
conservatism. Recently, a new adaptive control barrier function (aCBF) method based on a batch least-squares identifier (BaLSI)
\cite{karafyllis2018adaptive,karafyllis2019adaptivea} was proposed in \cite{wang2025output}, which does not impose extra restrictions on initial states beyond
the safe set and maintains the states in the original safe region after a
finite time. Based on this framework, the present paper extends the output-positive adaptive control design in \cite{wang2025output} from hyperbolic PDE systems to parabolic PDE--ODE cascades.

\subsection{Main contribution}
\begin{enumerate}[label=\arabic*)]
 \item Compared with existing adaptive boundary control for parabolic PDEs
 \cite{smyshlyaev2010adaptive,krstic2008adaptive,smyshlyaev2007adaptive,li2020adaptive,li2024condition}, the control design in this work further
 provides rigorous safety guarantees on the output.
 \item While \cite{koga2023safe} focuses on safe backstepping control for a
 Stefan model described by a parabolic PDE, this work addresses a broader class
 of problems by incorporating in-domain instabilities and parametric
 uncertainties.
 \item In contrast to the output-positive adaptive control for hyperbolic
 PDE-ODE cascades presented in \cite{wang2025output}, this paper considers
 parabolic PDE-ODE systems. To the best of our knowledge, this is the first
 result about safe adaptive control for parabolic PDE-ODE cascades.
\end{enumerate}

\subsection{Notation}
The symbol $\mathbb Z_{+}$ denotes the set of nonnegative integers. For
$f\in L^2(0,1)$, let $\|f\|^2=\int_0^1f(x)^2dx$. For $u(x,t)$, write $u[t]$
for the profile $u[t](x)=u(x,t)$. The Euclidean norm of a
vector and the induced matrix norm are denoted by $|\cdot|$. For a symmetric
matrix $Q$, $\lambda_{\min}(Q)$ denotes its smallest eigenvalue. Define
$\underline y_i=[y_1,\ldots,y_i]^\top$,
$\omega_j=\frac{\pi}{2}+j\pi$, and let $C_i$ be the row selector satisfying
$C_iY=y_i$.

\section{Problem Formulation}
The considered plant is
\begin{align}
    \dot Y(t)&=AY(t)+Bu(0,t),\label{eq:plant1}\\
    u_t(x,t)&=\varepsilon u_{xx}(x,t)+\lambda u(x,t),\label{eq:plant2}\\
    u_x(0,t)&=0,\label{eq:plant3}\\
    u(1,t)&=U(t).\label{eq:plant4}
\end{align}
for $(x,t)\in[0,1]\times[0,\infty)$. Here $n\geq2$. The function $U(t)$ is the
control input to be designed and $u(x,t)\in\mathbb R$ is the state of the PDE,
$Y(t)=[y_1,\ldots,y_n]^\top\in\mathbb R^n$ is the state of the ODE. The matrix
$A$ and the column vector $B$ are
\begin{equation}
 A=\begin{pmatrix}
 0&1&0&0&\cdots&0\\
 0&0&1&0&\cdots&0\\
 &&\vdots&&&\\
 0&0&0&0&\cdots&1\\
 I_1&I_2&I_3&\cdots&I_{n-1}&I_n
 \end{pmatrix},\qquad
 B=\begin{pmatrix}0\\0\\\vdots\\0\\b\end{pmatrix}.
\end{equation}
The constants $I_1,\ldots,I_n$ and $\varepsilon>0$ are known. This indicates
that the $Y$-ODE is in the controllable form, which covers many practical
models. The column vector $B=[0,0,\cdots,b]^\top$, where, without loss of
generality, we consider $b>0$. 
\textbf{Control objective.} The goal is twofold: the overall system
\eqref{eq:plant1}--\eqref{eq:plant4}, i.e., the $u$-PDE and the $Y$-ODE, is to be
exponentially regulated, and \emph{safety} is to be enforced, defined here as
the nonnegativity (without loss of generality for not exceeding a nonzero
setpoint) of the output of the distal ODE, i.e., the state furthest from the
control input:
\begin{equation}\label{eq:safety_def}
 y_1(t)\geq0,\qquad t\in[0,\infty).
\end{equation}
Hence the safe region is the one above the zero line.
For the nominal design,
$\lambda$ and $b$ are known. For the adaptive design, they are unknown and
satisfy Assumption~\ref{asmp:bounds} below.
\begin{assumption}\label{asmp:initial}
The output of the system is initially positive, i.e., $y_1(0)>0$.
\end{assumption}
\begin{assumption}\label{asmp:bounds}
The bounds of the unknown parameters $\lambda$ and $b$ are known and arbitrary,
i.e., $\underline\lambda\leq\lambda\leq\overline\lambda$ and
$0<\underline b\leq b\leq\overline b$.
\end{assumption}
For the adaptive design write $\theta=[\lambda,b]^\top$ and define
$ \Theta=\{\ell\in\mathbb R^2:\underline\lambda\leq\ell_1\leq\overline\lambda,
 \ \underline b\leq\ell_2\leq\overline b\}$.

\section{Nominal Safe Control}
\subsection{Transformations for the distal ODE}
To preserve the sign of the output $y_1$, we apply the non-overshooting
backstepping transformation
\begin{align}
 z_1(t)&=y_1(t),\label{eq:z1}\\
 z_i(t)&=y_i(t)-g_{i-1}(\underline y_{i-1}(t)),\qquad i=2,\ldots,n,\label{eq:zi}
\end{align}
and
\begin{align}
 g_0&=0,\label{eq:g0}\\
 g_i(\underline y_i(t))&=-\kappa_i z_i(t)
 +\sum_{j=1}^{i-1}\frac{\partial g_{i-1}}{\partial y_j}y_{j+1},
 \qquad i=1,\ldots,n-1.\label{eq:gi}
\end{align}
The positive design parameters $\kappa_i$ will be chosen later. Applying the transformation \eqref{eq:z1}--\eqref{eq:zi}, we now convert the
$Y$-ODE in \eqref{eq:plant1} into
\begin{equation}\label{eq:z_dyn}
 \dot Z(t)=A_zZ(t)+Bu(0,t)+BK^\top Y(t),
\end{equation}
where
\begin{equation}\label{eq:Az}
 A_z=\begin{pmatrix}
 -\kappa_1&1&0&0&\cdots&0\\
 0&-\kappa_2&1&0&\cdots&0\\
 &&\vdots&&&\\
 0&0&0&-\kappa_{n-1}&\cdots&1\\
 0&0&0&0&\cdots&-\kappa_n
 \end{pmatrix}
\end{equation}
is Hurwitz, and the constant row vector is
\begin{equation}\label{eq:KT}
\begin{aligned}
 K^\top=\frac{1}{b}\Big[&I_1-\kappa_n\frac{\partial g_{n-1}}{\partial y_1},\ 
 I_2-\frac{\partial g_{n-1}}{\partial y_1}-\kappa_n\frac{\partial g_{n-1}}{\partial y_2},\ \ldots,\\
 &I_n-\frac{\partial g_{n-1}}{\partial y_{n-1}}+\kappa_n\Big]_{1\times n}.
\end{aligned}
\end{equation}
Since the plant is linear, $g_i(\underline{y}_i)$ in \eqref{eq:zi} can be expressed as a linear combination of $y_1,\cdots,y_i$. The transformation \eqref{eq:z1}, \eqref{eq:zi} can be written as $Z=T_zY$ where $T_z$ is a nonsingular
lower-triangular constant matrix with unit diagonal.

\subsection{PDE Backstepping Transformation}\label{sec:pde_backstepping}
In order to remove the destabilizing term from the PDE and cancel the extra
terms in \eqref{eq:z_dyn}, we introduce the following backstepping
transformation:
\begin{equation}\label{eq:w_transform}
 w(x,t)=u(x,t)-\int_0^xk(x,y)u(y,t)dy-r(x)Y(t).
\end{equation}
with $k$ and $r$ satisfying
\begin{align}
 \varepsilon k_{xx}(x,y)&=\varepsilon k_{yy}(x,y)+(\lambda+c)k(x,y),\label{eq:kernel_k}\\
 k(x,x)&=-\frac{\lambda+c}{2\varepsilon}x,\label{eq:kernel_diag}\\
 k_y(x,0)&=-r(x)\frac{B}{\varepsilon},\label{eq:kernel_ky}\\
 \varepsilon r''(x)&=r(x)(A+cI),\label{eq:kernel_r}\\
 r(0)&=-K^\top,\qquad r'(0)=0\label{eq:kernel_rbc}
\end{align}
where $c$ is a positive constant. The solutions of \eqref{eq:kernel_k}--\eqref{eq:kernel_rbc} are given in Appendix~\ref{appendix:kernels}. With
\eqref{eq:w_transform}, the original system \eqref{eq:plant2}--\eqref{eq:plant4} is converted into the following system
\begin{align}
 \dot Z(t)&=A_zZ(t)+Bw(0,t),\label{eq:target_ode}\\
 w_t(x,t)&=\varepsilon w_{xx}(x,t)-cw(x,t),\label{eq:target_pde}\\
 w_x(0,t)&=0,\label{eq:target_neumann}\\
 w(1,t)&=\delta(t).\label{eq:target_dirichlet}
\end{align}
by choosing the control input as
\begin{equation}\label{eq:input}
 U(t)=\int_0^1k(1,y)u(y,t)dy+r(1)Y(t)+\delta(t),
\end{equation}
where, as will be seen more clearly later, $\delta(t)$ is introduced to
enforce the output safety, and is designed as the following form
\begin{equation}\label{eq:delta}
    \delta(t)=Me^{-ct},\qquad M\geq0.
\end{equation}
 where $M$ will be selected later. The target PDE solution is
\begin{equation}\label{eq:w_solution}
\begin{split}
 w(x,t) &= \delta(t)+\sum_{j=0}^{\infty}\cos(\omega_j x)\Bigg[ e^{-(\varepsilon\omega_j^2+c)t}\theta_j \\
 &\quad -\frac{2(-1)^j}{\omega_j}\int_0^t e^{-(\varepsilon\omega_j^2+c)(t-\tau)}
 \bigl(c\delta(\tau)+\dot\delta(\tau)\bigr)d\tau \Bigg], \\
\end{split}
\end{equation}
where
$\theta_j=2\int_0^1[w(x,0)-\delta(0)]\cos(\omega_jx)dx$.

\subsection{Selection of safe design parameters}
According to Assumption \ref{asmp:initial}, we know 
\begin{equation}\label{eq:z1_positive} z_1(0)=y_1(0)>0.\end{equation} 
According to \eqref{eq:z1}--\eqref{eq:gi}, we have
\begin{align}\label{eq:z_i_initial}
 z_{i+1}(0)&=y_{i+1}(0)-g_i(\underline y_i(0))\notag\\
 &=y_{i+1}(0)+\kappa_i z_i(0)-\sum_{j=1}^{i-1}
 \frac{\partial g_{i-1}}{\partial y_j}y_{j+1}(0)\notag\\
 &=\left[C_{i+1}-\sum_{j=1}^{i-1}
 \frac{\partial g_{i-1}}{\partial y_j}C_{j+1}\right]Y(0)+\kappa_i z_i(0).
\end{align}
The design parameters $\kappa_i$, $i=1,\ldots,n-1$, are selected to satisfy
\begin{equation}\label{eq:kappa_i}
 \kappa_i>\frac{\left[\displaystyle\sum_{j=1}^{i-1}
 \frac{\partial g_{i-1}}{\partial y_j}C_{j+1}-C_{i+1}\right]Y(0)}
 {z_i(0)},\quad i=1,\ldots,n-1.
\end{equation}

The gain condition \eqref{eq:kappa_i} ensures the following lemma.
\begin{lemma}\label{lemma:positive_initial}
With the design parameters $\kappa_i$, $i=1,\ldots,n-1$, satisfying
\eqref{eq:kappa_i}, the non-overshooting states are
initialized positive, i.e., $z_i(0)>0$ for $i=1,\ldots,n$.
\end{lemma}
\begin{proof}
According to \eqref{eq:z1_positive}, we know $z_1(0)>0$. Recalling
\eqref{eq:z_i_initial} and choosing $\kappa_i$ to satisfy \eqref{eq:kappa_i},
we can obtain that $z_{i+1}(0)>0$. Through the recursive process, we obtain
that $z_i(0)>0$ for $i=1,\ldots,n$. The proof is complete.
\end{proof}

According to \eqref{eq:target_ode}, we have
\begin{equation}\label{eq:z_n}
 z_n(t)=e^{-\kappa_nt}z_n(0)+e^{-\kappa_nt}\int_0^te^{\kappa_n\tau}bw(0,\tau)d\tau.
\end{equation}
 Set
\begin{equation}\label{eq:kappa_n_c}
    \kappa_n=c,
\end{equation}
With $\delta(t)=Me^{-ct}$ in \eqref{eq:delta}, the integral term in \eqref{eq:w_solution}
vanishes, and
\begin{equation}\label{eq:D}
 D(t):=1-2\sum_{j=0}^{\infty}e^{-\varepsilon\omega_j^2t}\frac{(-1)^j}{\omega_j}
\end{equation}
satisfies the following property.
\begin{lemma}\label{lemma:D}
For all $t>0$, the function \eqref{eq:D} satisfies $0<D(t)<1$.
\end{lemma}\begin{proof}
The proof of Lemma~\ref{lemma:D} is given in Appendix~\ref{appendix:D}.
\end{proof}
Substituting
\eqref{eq:delta} into \eqref{eq:w_solution} gives
\begin{equation}\label{eq:w0_nonneg}
 w(0,t)=e^{-ct}\Bigg[M D(t)
 +2\sum_{j=0}^{\infty}e^{-\varepsilon\omega_j^2t}\tilde\theta_j\Bigg],
\end{equation}
where $\tilde\theta_j=\int_0^1w(x,0)\cos(\omega_jx)dx$. Combining this with
\eqref{eq:z_n}, we obtain
\begin{equation}\label{eq:decomp_nom}
 z_n(0)+b\int_0^t e^{c\tau}w(0,\tau)d\tau
 =z_n(0)+bM\int_0^t D(\tau)d\tau+\mathcal I(t),
\end{equation}
where $\mathcal I(t)=2b\sum_{j=0}^{\infty}\tilde\theta_j
\int_0^t e^{-\varepsilon\omega_j^2\tau}d\tau$. Define
\begin{equation}\label{eq:S_def}
 S(t):=\sum_{j=0}^{\infty}\frac{1-e^{-\varepsilon\omega_j^2t}}{\omega_j^2},
\end{equation}
whose property is given in the following lemma. 
\begin{lemma}\label{lemma:S}
The series in \eqref{eq:S_def} converges uniformly on $[0,\infty)$, and $S$ is continuous on $[0,\infty)$. Moreover, $S$ is strictly increasing, with $S(0)=0$ and
$\lim_{t\to\infty}S(t)=1/2$, i.e., $0\le S(t)<1/2$.
\end{lemma}
\begin{proof}
For every $t\ge0$,
\[
 0\le\frac{1-e^{-\varepsilon\omega_j^2t}}{\omega_j^2}\le\frac{1}{\omega_j^2},
 \qquad
 \sum_{j=0}^{\infty}\frac{1}{\omega_j^2}=\frac12.
\]
The Weierstrass test therefore gives uniform convergence on $[0,\infty)$, so
$S\in C([0,\infty))$. Every summand is strictly increasing in $t$, so $S$ is
strictly increasing and $S(0)=0$. As $t\to\infty$ each summand increases to
$\omega_j^{-2}$, and monotone convergence gives
$\lim_{t\to\infty}S(t)=\sum_{j=0}^{\infty}\omega_j^{-2}=1/2$.
\end{proof}

Choose the constant $M$ in $\delta(t)=Me^{-ct}$ as follows:
\begin{equation}\label{eq:M_check}
\begin{aligned}
 M&>
 \begin{cases}
 0,
 &~if~
 \dfrac{\|w(\cdot,0)\|}{\sqrt{2}\varepsilon}
 \le \dfrac{z_n(0)}{b},
 \\[10pt]
 \dfrac{
 \dfrac{\|w(\cdot,0)\|}{\sqrt{2}\varepsilon}
 -\dfrac{z_n(0)}{b}
 }{
 \displaystyle\int_0^{\tau_M}D(\tau)d\tau
 },
 &~if~
 \dfrac{\|w(\cdot,0)\|}{\sqrt{2}\varepsilon}
 >\dfrac{z_n(0)}{b},
 \end{cases}
\end{aligned}
\end{equation}
where $\tau_M>0$ is the solution of
\begin{equation}\label{eq:tau_M_nom}
 S(\tau_M)=
 \frac{\varepsilon z_n(0)}{\sqrt{2}b\|w(\cdot,0)\|}
\end{equation}
when $ \dfrac{\|w(\cdot,0)\|}{\sqrt{2}\varepsilon}
 >\dfrac{z_n(0)}{b}$, with $S$ defined in \eqref{eq:S_def}.
 
\begin{lemma}\label{lemma:tau_M_nom}
The equation \eqref{eq:tau_M_nom} has a unique bounded
solution $\tau_M>0$ when $\dfrac{\|w(\cdot,0)\|}{\sqrt{2}\varepsilon}
 >\dfrac{z_n(0)}{b}$.
\end{lemma}
\begin{proof}
Considering $z_n(0)>0$ by Lemma~\ref{lemma:positive_initial}, when $\dfrac{\|w(\cdot,0)\|}{\sqrt{2}\varepsilon}
 >\dfrac{z_n(0)}{b}$, the right-hand side of \eqref{eq:tau_M_nom}
lies in $(0,1/2)$. By Lemma~\ref{lemma:S}, $S(t)$ increases strictly from $S(0)=0$
to $\lim_{t\to\infty}S(t)=1/2$, so the intermediate value theorem gives a unique bounded $\tau_M>0$ solution to \eqref{eq:tau_M_nom}.
\end{proof}
\subsection{Result of the nominal safe control}
\begin{theorem}\label{theorem:nominal}
For initial condition $u[0]\in L^2(0,1)$ and $Y(0)\in\mathbb R^n$ satisfying
Assumption~\ref{asmp:initial}, if the design parameters $\kappa_i$,
$i=1,\ldots,n-1$, satisfy \eqref{eq:kappa_i}, $\kappa_n$ satisfies
\eqref{eq:kappa_n_c}, and $M$ satisfies \eqref{eq:M_check}, the closed-loop
system consisting of the plant \eqref{eq:plant1}--\eqref{eq:plant4} and the
control law \eqref{eq:input} have the following properties:
\begin{enumerate}[label=\arabic*)]
 \item All states are convergent to zero, i.e., $\lim_{t\to\infty}(\|u_x(\cdot,t)\|^2+\|u(\cdot,t)\|^2+|Y(t)|^2)=0$, and the convergence is exponential.
 \item Safety is ensured in the sense that $y_1(t)\geq0$ for all $t\geq0$.
 \item The control input $U(t)$ converges exponentially to zero.
\end{enumerate}
\end{theorem}
\begin{proof}
(1) We show that the convergence to zero of all states in property 1 is
achieved by Lyapunov analysis. Defining the following transformation:
\begin{equation}\label{eq:v_transform}
 w(x,t)=v(x,t)+\delta(t),
\end{equation}
Considering \eqref{eq:delta}, we have
\begin{align}
 \dot Z(t)&=A_zZ(t)+Bv(0,t)+B\delta(t),\label{eq:v_ode}\\
 v_t(x,t)&=\varepsilon v_{xx}(x,t)-cv(x,t),\label{eq:v_pde}\\
 v_x(0,t)&=0,\qquad v(1,t)=0.\label{eq:v_bc}
\end{align}
Define the Lyapunov function
\begin{equation}\label{eq:V}
 V(t)=Z(t)^\top PZ(t)+\frac{1}{2}\|v(\cdot,t)\|^2
 +\frac{a_1}{2}\|v_x(\cdot,t)\|^2+\frac{a_2}{2}\delta(t)^2,
\end{equation}
where the positive definite matrix $P=P^\top$ is the solution of the
Lyapunov equation $A_z^\top P+PA_z=-Q$ for some $Q=Q^\top>0$, where
$q=\lambda_{\min}(Q)$, and where the
positive analysis parameters $a_1$ and $a_2$ are to be determined later.
Defining
\begin{equation}\label{eq:Xi}
 \Xi(t)=|Z(t)|^2+\|v(\cdot,t)\|^2+\|v_x(\cdot,t)\|^2+\delta(t)^2,
\end{equation}
we have
\begin{equation}\label{eq:Xi_bound}
 \xi_1\Xi(t)\leq V(t)\leq\xi_2\Xi(t)
\end{equation}
for some positive $\xi_1$ and $\xi_2$.
Taking the time derivative of $V(t)$ along \eqref{eq:v_ode}--\eqref{eq:v_bc},
applying Young's inequality and Cauchy--Schwarz inequality, we have
\begin{align}\label{eq:dV1}
 &\dot V(t)=-Z^\top QZ+2Z^\top PBv(0,t)+2Z^\top PB\delta(t)\notag\\
 &-\varepsilon\|v_x\|^2-c\|v\|^2-a_1\varepsilon\|v_{xx}\|^2
 -a_1c\|v_x\|^2-a_2c\delta^2\notag\\
 &\leq-\frac{q}{3}|Z|^2-c\|v\|^2
 -\left(\varepsilon+a_1c-\frac{3|PB|^2}{q}\right)\|v_x\|^2\notag\\
 &\quad-\left(a_2c-\frac{3|PB|^2}{q}\right)\delta^2,
\end{align}
where we used $v_x(0,t)=0$, $v(1,t)=0$ in the
boundary terms of the integration by parts, $v_{xx}(1,t)=0$ which follows from
$v(1,t)=0$ and \eqref{eq:v_pde} at $x=1$, the trace estimate
$v(0,t)^2\leq\|v_x(\cdot,t)\|^2$ (which follows from $v(1,t)=0$ and
Cauchy--Schwarz), and Young's inequality. Choosing $a_1$ and $a_2$ large enough
to satisfy
\begin{equation}\label{eq:a_choice}
 a_1>\max\left\{\frac{3|PB|^2}{cq}-\frac{\varepsilon}{c},0\right\},\qquad
 a_2>\frac{3|PB|^2}{cq},
\end{equation}
we then obtain that
\begin{equation}\label{eq:Vdot}
 \dot V(t)\leq-\gamma V(t),
\end{equation}
where
\begin{equation}\label{eq:gamma}
\gamma=\frac{1}{\xi_2}\min\left\{\frac{q}{3},
\varepsilon+a_1c-\frac{3|PB|^2}{q},c,
a_2c-\frac{3|PB|^2}{q}\right\}>0.
\end{equation}
Recalling \eqref{eq:Xi_bound}, it follows that
\begin{equation}\label{eq:exp_decay}
\Xi(t)\leq\frac{\xi_2}{\xi_1}\Xi(0)e^{-\gamma t}.
\end{equation}
Next, we show the convergence to zero of the original states from the target
system's stability obtained from \eqref{eq:exp_decay}. From \eqref{eq:exp_decay}
and \eqref{eq:Xi}, we directly have that $|Z(t)|$, $\|v(\cdot,t)\|$,
$\|v_x(\cdot,t)\|$ and $\delta(t)$ converge to zero. According to the
transformation $Z=T_zY$, the convergence of $|Y(t)|$ is obtained. The inverse of
\eqref{eq:w_transform} is obtained as
\begin{equation}\label{eq:inv_transform}
 u(x,t)=w(x,t)+\int_0^x\bar k(x,y)w(y,t)dy+\bar r(x)Y(t),
\end{equation}
where the kernels $\bar k(x,y)$ and $\bar r(x)$ can be obtained similarly to
those in the transformation \eqref{eq:w_transform} and are therefore bounded.
Using the convergence to zero of $\|w(\cdot,t)\|$ and $\|w_x(\cdot,t)\|$
obtained from the convergence of $\|v(\cdot,t)\|$, $\|v_x(\cdot,t)\|$,
$\delta(t)$ and $|Y(t)|$, we obtain the convergence of $\|u(\cdot,t)\|$ and
$\|u_x(\cdot,t)\|$ via \eqref{eq:inv_transform}. The property 1 is obtained.

(2) According to \eqref{eq:target_ode}, we have \eqref{eq:z_n}. It is obvious
that $z_n(t)>0$ if $z_n(0)+b\int_0^t e^{\kappa_n\tau}w(0,\tau)d\tau>0$ for all
$t\ge0$. By Cauchy--Schwarz,
$\|\cos(\omega_j\cdot)\|=1/\sqrt2$, so $|\tilde\theta_j|\le\|w(\cdot,0)\|/\sqrt2$;
with $\int_0^t e^{-\varepsilon\omega_j^2\tau}d\tau
=(1-e^{-\varepsilon\omega_j^2t})/(\varepsilon\omega_j^2)$ and \eqref{eq:S_def},
the perturbation in \eqref{eq:decomp_nom} obeys
\begin{equation}\label{eq:I_bound_nom}
 \frac{|\mathcal I(t)|}{b}
 \le\frac{\sqrt2\|w(\cdot,0)\|}{\varepsilon}S(t).
\end{equation}
By \eqref{eq:decomp_nom} and \eqref{eq:I_bound_nom}, the bracket in
\eqref{eq:z_n} satisfies, for $t>0$,
\begin{equation}\label{eq:safety_bracket_positive_nom}
\begin{split}
 z_n(0)+&b\int_0^t e^{c\tau}w(0,\tau)d\tau
 = z_n(0)+bM\int_0^t D(\tau)d\tau+\mathcal I(t)\\
 &\ge b\Big[\frac{z_n(0)}{b}+M\int_0^t D(\tau)d\tau
 -\frac{\sqrt2\|w(\cdot,0)\|}{\varepsilon}S(t)\Big].
\end{split}
\end{equation}
If $\|w(\cdot,0)\|/(\sqrt2\varepsilon)\le z_n(0)/b$, then $S(t)<1/2$ by
Lemma~\ref{lemma:S}, so
$\frac{\sqrt2\|w(\cdot,0)\|}{\varepsilon}S(t)
<\frac{\|w(\cdot,0)\|}{\sqrt2\varepsilon}\le\frac{z_n(0)}{b}$, and the bracket
in \eqref{eq:safety_bracket_positive_nom} is positive. Otherwise
$\dfrac{\|w(\cdot,0)\|}{\sqrt2\varepsilon}>\dfrac{z_n(0)}{b}$ holds and $\tau_M$ is finite by
Lemma~\ref{lemma:tau_M_nom}. For $0<t\le\tau_M$, the monotonicity of $S$ and
\eqref{eq:tau_M_nom} give
$\frac{\sqrt2\|w(\cdot,0)\|}{\varepsilon}S(t)\le\frac{z_n(0)}{b}$, so the bracket
is positive since $M>0$ and $D(\tau)>0$ on $(0,\infty)$; and for $t\ge\tau_M$,
using $\int_0^tD\ge\int_0^{\tau_M}D$ and $S(t)<1/2$, the second branch of
\eqref{eq:M_check} gives
$M\int_0^tD>\frac{\|w(\cdot,0)\|}{\sqrt2\varepsilon}-\frac{z_n(0)}{b}
\ge\frac{\sqrt2\|w(\cdot,0)\|}{\varepsilon}S(t)-\frac{z_n(0)}{b}$, so the
bracket is again positive. By Lemma~\ref{lemma:positive_initial}, $z_n(0)>0$, so
$z_n(t)=e^{-ct}\big(z_n(0)+b\int_0^t e^{c\tau}w(0,\tau)d\tau\big)>0$ for all
$t\geq0$. For $i=n-1,\ldots,1$, the target equations give
\begin{equation}\label{eq:z_i_safety}
 z_i(t)=e^{-\kappa_i t}z_i(0)+
 \int_0^t e^{-\kappa_i(t-\tau)}z_{i+1}(\tau)d\tau>0.
\end{equation}
Since $y_1=z_1$, this proves $y_1(t)\geq0$ for all $t\geq0$, which is
property 2.

(3) Finally, the bounded kernels and the exponential convergence of the
states, together with $\delta(t)=Me^{-ct}$ in \eqref{eq:input}, give the
exponential convergence of $U(t)$ to zero, which is property 3. The proof of
this theorem is complete.
\end{proof}

\begin{remark}\label{remark:tradeoff}
The design parameter $c$ is the target decay rate: increasing $c$ accelerates
the convergence of all plant states, but it also increases the control
gains considering $k(x,x)=-(\lambda+c)x/(2\varepsilon)$. The margin amplitude $M$ must satisfy \eqref{eq:M_check} to preserve
$y_1(t)\geq0$. A larger
$M$ enlarges the safety margin at the cost of a larger control input, due to
$\delta(t)=Me^{-ct}$, which decays to zero and does not affect the final exponential
regulation.
\end{remark}

\section{Safe Adaptive Control Design}
\subsection{Batch least-squares identifier}
The parameter estimate $\hat\theta=[\hat\lambda,\hat b]^\top$ is generated by a
triggered batch least-squares identifier, which is built in this section as
follows.
Let $T>0$ be the sampling period and define the time sequence $t_{i+1}=t_i+T$.

According to \eqref{eq:plant1} and \eqref{eq:plant2}, similar to \cite{wang2022event}, we obtain for $\tau>0$ and $\bar n=1,2,\ldots$ that
\begin{align}
 &\frac{d}{d\tau}\int_0^1\sin(x\pi\bar n)u(x,\tau)dx
 = -\varepsilon\pi\bar n(-1)^{\bar n}u(1,\tau)+\varepsilon\pi\bar n u(0,\tau)\notag\\
 &\qquad\qquad-\varepsilon(\pi\bar n)^2\int_0^1\sin(x\pi\bar n)u(x,\tau)dx\notag\\
 &\qquad\qquad+\lambda\int_0^1\sin(x\pi\bar n)u(x,\tau)dx,\label{eq:modal}\\
 &\dot y_n(\tau)=\sum_{i=1}^nI_iy_i(\tau)+bu(0,\tau).\label{eq:last_ode}
\end{align}
Define the instant $\mu_{i+1}$ as
\begin{equation}\label{eq:mu}
 \mu_{i+1}=\min\{t_j:j\in\{0,\ldots,i\},\ t_j\geq t_{i+1}-\widetilde NT\},
\end{equation}
where $\widetilde N\geq1$ is a positive integer. Both $T$ and
$\widetilde N$ are free design parameters. Increasing $T$ or $\widetilde N$
enlarges the data window, so more measured data take part in the estimation,
which improves the estimation accuracy but prolongs the adaptive learning
phase. Integrating \eqref{eq:modal} and \eqref{eq:last_ode} from
$\mu_{i+1}$ to $t$, we obtain
\begin{align}
 p_{\bar n}(t,\mu_{i+1})&=\lambda g_{\bar n}(t,\mu_{i+1}),\label{eq:p_g}\\
 p_b(t,\mu_{i+1})&=bq_b(t,\mu_{i+1}).\label{eq:p_b}
\end{align}
where $p_{\bar n}$, $g_{\bar n}$, $p_b$ and $q_b$ are
\begin{align}
 &p_{\bar n}(t,\mu_{i+1})=\int_0^1\sin(x\pi\bar n)u(x,t)dx\notag\\
 &+\varepsilon\pi\bar n(-1)^{\bar n}
 \int_{\mu_{i+1}}^t u(1,\tau)d\tau
 -\varepsilon\pi\bar n\int_{\mu_{i+1}}^t u(0,\tau)d\tau\notag\\
 &+\varepsilon(\pi\bar n)^2\int_{\mu_{i+1}}^t\int_0^1
 \sin(x\pi\bar n)u(x,\tau)dxd\tau\notag\\
 &-\int_0^1\sin(x\pi\bar n)u(x,\mu_{i+1})dx,\label{eq:p_def}\\
 g_{\bar n}(t,\mu_{i+1})&=\int_{\mu_{i+1}}^t\int_0^1
 \sin(x\pi\bar n)u(x,\tau)dxd\tau,\label{eq:g_def}
\end{align}
for $\bar n=1,2,\ldots$, and
\begin{align}
 p_b(t,\mu_{i+1})&=y_n(t)-y_n(\mu_{i+1})
 -\int_{\mu_{i+1}}^t\sum_{i=1}^nI_i y_i(\tau)d\tau,\label{eq:p_b_def}\\
 q_b(t,\mu_{i+1})&=\int_{\mu_{i+1}}^tu(0,\tau)d\tau.\label{eq:q_b_def}
\end{align}
Define a function $h_{i,\bar n}:\mathbb R^3\to\mathbb R_+$ by the formula
\begin{equation}\label{eq:cost}
\begin{aligned}
 h_{i,\bar n}(\ell)=&\int_{\mu_{i+1}}^{t_{i+1}}
 \Big[\big(p_{\bar n}(t,\mu_{i+1})-\ell_1 g_{\bar n}(t,\mu_{i+1})\big)^2\\
 &\quad+\big(p_b(t,\mu_{i+1})-\ell_2 q_b(t,\mu_{i+1})\big)^2\Big]dt.
\end{aligned}
\end{equation}
According to \eqref{eq:p_g}, the function $h_{i,\bar n}(\ell)$ has a global
minimum $h_{i,\bar n}(\theta)=0$. We get from Fermat's theorem (vanishing
gradient at extrema), that is, differentiating the function $h_{i,\bar n}(\ell)$
defined by \eqref{eq:cost} with respect to $\ell_1$ and $\ell_2$ and setting the
derivatives at the position of the global minimum $[\ell_1,\ell_2]=[\lambda,b]$
to zero, that the following matrix equation holds for every $i$ and
$\bar n=1,2,\ldots$:
\begin{equation}\label{eq:consistency}
 Z_{\bar n}(\mu_{i+1},t_{i+1})=G_{\bar n}(\mu_{i+1},t_{i+1})\theta,
\end{equation}
where $\theta=[\lambda,b]^\top$ is a column vector of the unknown parameters,
and where
\begin{equation}
 Z_{\bar n}=\begin{bmatrix}H_{\bar n,1}\\H_2\end{bmatrix},\qquad
 G_{\bar n}=\begin{bmatrix}Q_{\bar n,1}&0\\0&Q_2\end{bmatrix},
\end{equation}
with $H_{\bar n,1}$, $H_2$, $Q_{\bar n,1}$ and $Q_2$ defined as
\begin{align}
 H_{\bar n,1}(\mu_{i+1},t_{i+1})
 &=\int_{\mu_{i+1}}^{t_{i+1}}p_{\bar n}(t,\mu_{i+1})g_{\bar n}(t,\mu_{i+1})dt,\label{eq:H1}\\
 H_2(\mu_{i+1},t_{i+1})
 &=\int_{\mu_{i+1}}^{t_{i+1}}p_b(t,\mu_{i+1})q_b(t,\mu_{i+1})dt,\label{eq:H2}\\
 Q_{\bar n,1}(\mu_{i+1},t_{i+1})
 &=\int_{\mu_{i+1}}^{t_{i+1}}g_{\bar n}^2(t,\mu_{i+1})dt,\label{eq:Q1}\\
 Q_2(\mu_{i+1},t_{i+1})
 &=\int_{\mu_{i+1}}^{t_{i+1}}q_b^2(t,\mu_{i+1})dt.\label{eq:Q2}
\end{align}
The parameter estimator (update law) is defined as
\begin{align}\label{eq:update}
 \hat\theta(t_{i+1})=\argmin\{&|\ell-\hat\theta(t_i)|^2:\ell\in\Theta,\notag\\
 &Z_{\bar n}=G_{\bar n}\ell,\quad \bar n=1,2,\ldots\}.
\end{align}

\subsection{Safe Adaptive Controller}
The safe adaptive controller $U_a(t)$ is designed as 
\begin{equation}\label{eq:Ua}
 U_a(t,\check\theta(t))=\int_0^1\check k(1,y,\check\theta(t))u(y,t)dy+\bar r(1,\check\theta(t))Y(t)+\delta(t),
\end{equation}
where $\delta(t)=Me^{-ct}$ with the positive constant $M$
that will be shown later, and where $\check k$ and $\check r$ are obtained from \eqref{eq:kernel_k} and
\eqref{eq:kernel_r} after replacing $(\lambda,b)$ by a time-varying gain
$\check\theta(t)\in\Theta$ that is 
\begin{equation}\label{eq:theta_max}
 \check\theta(t)=
 \begin{cases}
  \displaystyle\argmax_{\ell\in\Theta}
  \Big[\int_0^1 k(1,y;\ell)u(y,t)dy+r(1;\ell)Y(t)\Big],\\
   \qquad \qquad \quad t\in[0,t_1),\\[10pt]
  \hat\theta(t_i), \qquad t\in[t_i,t_{i+1}),\quad i\ge1,
 \end{cases}
\end{equation}
where $\hat\theta(t_i)$ is defined by \eqref{eq:update}. The first item in \eqref{eq:theta_max}
define the controller gain on the $[0,t_1)$ for the purpose of safety; the second item in \eqref{eq:theta_max} is the identifier
\eqref{eq:update} that runs in parallel with $\hat\theta(t_0)$ serving as its
initialization, and whose output is used from $t_1$ onward. 
\begin{lemma}\label{lemma:identification}
Under the design parameter
\begin{align}
  M\neq-\check r(1;\check\theta(0))Y(0),\label{eq:M1}
\end{align}
the finite-time exact identification of the unknown parameters is achieved at
the first update time $t_1$, i.e., $\hat\theta(t)=\theta$ for all $t\geq t_1$.
\end{lemma}
\begin{proof}
We first prove by contradiction that $u(0,t)\not\equiv0$ on $[0,t_1]$. If
$u(0,t)\equiv0$ on $[0,t_1]$, then \eqref{eq:plant3} also gives
$u_x(0,t)=0$, so $u$ has zero Cauchy data on the line $\{x=0\}$. Since this
line is non-characteristic for the heat operator in \eqref{eq:plant2} (its
conormal symbol equals $\varepsilon\neq0$), the unique continuation property
of the heat equation \cite{saut1987unique} implies that $u$ vanishes
identically, i.e., $u(x,t)\equiv0$ for all $x\in[0,1]$ and $t\in(0,t_1]$
(and hence at $t=0$ by continuity), which implies $u(1,t)=U_a(t)\equiv0$ on
$[0,t_1]$. Recalling
\eqref{eq:Ua}, \eqref{eq:M1} and \eqref{eq:delta}, this yields
$U_a(0)=\check r(1;\check\theta(0))Y(0)+M\neq0$, a contradiction. Hence
$u(0,t)\not\equiv0$ on $[0,t_1]$. By \eqref{eq:q_b_def}, this implies
$Q_2>0$. Moreover, if $g_{\bar n}(t,\mu_1)$ were identically zero for every
$\bar n$, then all sine modal coefficients of $u(\cdot,t)$ would be zero on
$[0,t_1]$, which would imply $u(\cdot,t)=0$ and contradict
$u(0,t)\not\equiv0$. Thus there exists a mode $\bar n$ for which
$Q_{\bar n,1}>0$, and the diagonal matrix $G_{\bar n}$ is nonsingular. The
exact consistency relation \eqref{eq:consistency} then gives
$\hat\theta(t_1)=\theta$. Since $\theta\in\Theta$ remains feasible for every
later update, the distance-minimizing update preserves it, so
$\hat\theta(t)=\theta$ for all $t\geq t_1$.
\end{proof}
Defining $e(t)$ as the control input error,
\begin{align}\label{eq:e}
 e(t)=&U_a(t,\check\theta(t))-U_a(t,\theta)\notag\\
=&\int_0^1\big[\check k(1,y;\check\theta)-k(1,y;\theta)\big]u(y,t)dy\notag\\&
 +\big[\check r(1;\check\theta)-r(1;\theta)\big]Y(t).
\end{align}
Because $\check\theta(t)$ in \eqref{eq:theta_max} maximizes $U_a(t,\cdot)$ over
$\Theta$ and $\theta\in\Theta$, we obtain $U_a(t,\check\theta(t))\ge U_a(t,\theta)$,
hence
\begin{equation}\label{eq:e_positive}
 e(t)\ge0,\qquad t\in[0,t_1),\qquad e(t)=0,\qquad t\ge t_1,
\end{equation}
according to Lemma \ref{lemma:identification}.
Through the same process as in Section~\ref{sec:pde_backstepping}, the target
system is rewritten as
\begin{align}
 \dot Z(t)&=A_zZ(t)+Bw(0,t),\label{eq:target_ode_a}\\
 w_t(x,t)&=\varepsilon w_{xx}(x,t)-cw(x,t),\label{eq:target_pde_a}\\
 w_x(0,t)&=0,\label{eq:target_neumann_a}
\end{align}
and
\begin{equation}\label{eq:target_dirichlet_a}
 w(1,t)=\delta(t)+e(t),
\end{equation}
where, as shown in Appendix~\ref{appendix:w}, $w(x,t)$ is solved as, with
$d(t)=\delta(t)+e(t)$:
\begin{equation}\label{eq:w_solution_a}
\begin{aligned}
 &w(x,t)=d(t)+\sum_{j=0}^{\infty}\Bigg[e^{-(\varepsilon\omega_j^2+c)t}\chi_j\\
 &-\frac{2(-1)^j}{\omega_j}\int_0^t
 e^{-(\varepsilon\omega_j^2+c)(t-\tau)}
 \big(c d(\tau)+\dot d(\tau)\big)d\tau\Bigg]\cos(\omega_jx),
\end{aligned}
\end{equation}
where
\begin{equation}\label{eq:chi}
 \chi_j=2\int_0^1\big[w(x,0)-d(0)\big]\cos(\omega_jx)dx.
\end{equation}
According to \eqref{eq:target_ode_a}, $z_n(t)$ is 
\begin{equation}\label{eq:z_n_bracket}
 z_n(t)=e^{-ct}\Bigg[z_n(0)+b\int_0^t e^{c\tau}w(0,\tau)d\tau\Bigg].
\end{equation}
Substituting $d(t)=\delta(t)+e(t)=Me^{-ct}+e(t)$ into \eqref{eq:w_solution_a},
using $c\delta+\dot\delta=0$, and integrating the $e$-terms by parts gives
\begin{equation}\label{eq:w0_adaptive}
\begin{split}
 w(0,t) &= e^{-ct}\Bigg[ M D(t)+2\sum_{j=0}^{\infty}
 \tilde\theta_j e^{-\varepsilon\omega_j^2t} \\
 &\quad +2\varepsilon\sum_{j=0}^{\infty}(-1)^j\omega_j
 \int_0^t e^{-\varepsilon\omega_j^2(t-\tau)}e^{c\tau}e(\tau)d\tau\Bigg],
\end{split}
\end{equation}
where $\tilde\theta_j=\int_0^1w(x,0)\cos(\omega_jx)dx$ and $D(t)$ is defined
in \eqref{eq:D}.  Consequently,
\begin{equation}\label{eq:decomposition}
\begin{split}
 z_n(0)+b\int_0^t e^{c\tau}w(0,\tau)d\tau
 &= z_n(0)+bM\int_0^t D(\tau)d\tau \\
 &\quad +\mathcal I(t)+\mathcal E(t),
\end{split}
\end{equation}
where
\begin{align}
 \mathcal I(t)&=2b\sum_{j=0}^{\infty}\tilde\theta_j
 \int_0^t e^{-\varepsilon\omega_j^2\tau}d\tau,\label{eq:I_term}\\
 \mathcal E(t)&=2b\varepsilon\sum_{j=0}^{\infty}(-1)^j\omega_j
 \int_0^t e^{-\varepsilon\omega_j^2\tau}
 \int_0^\tau e^{(\varepsilon\omega_j^2+c)s}e(s)dsd\tau.\label{eq:E_term}
\end{align}
whose convergence is given in Appendix~\ref{appendix:series}. Using the fact
$\sum_{j=0}^{\infty}\frac{(-1)^j}{\omega_j}=\frac12$ and \eqref{eq:D}, then \eqref{eq:E_term} becomes
\begin{equation}\label{eq:E_convolution}
 \mathcal E(t)=b\int_0^t D(t-s)e^{cs}e(s)ds.
\end{equation}
Since $b>0$, $e^{cs}>0$, and $D(\cdot)>0$ on $(0,\infty)$ (see
Lemma~\ref{lemma:D}), the map $e\mapsto\mathcal E$ is order-preserving. Recalling
\eqref{eq:e_positive}, we obtain
\begin{equation}\label{eq:E_positive}
 \mathcal E(t)\ge0,\qquad t\ge0.
\end{equation}
By \eqref{eq:decomposition} and \eqref{eq:E_positive}, the content in the bracket in
\eqref{eq:z_n_bracket} satisfies
\begin{equation}\label{eq:bracket_lower}
 z_n(0)+b\int_0^t e^{c\tau}w(0,\tau)d\tau
 \ge z_n(0)+bM\int_0^t D(\tau)d\tau+\mathcal I(t).
\end{equation}
It therefore suffices to choose a constant $M$ for which the right-hand side
of \eqref{eq:bracket_lower} stays positive uniformly over the admissible
parameters. 

For each $\ell\in\Theta$, let $w_\ell(\cdot,0)$ denote the initial
target system state corresponding to the unknown plant parameters, and define
\begin{equation}\label{eq:W_S_definition}
 W:=\max_{\ell\in\Theta}\|w_\ell(\cdot,0)\|.
\end{equation}
The constant $M$ in $\delta(t)=Me^{-ct}$ is selected to satisfy both \eqref{eq:M1} and 
\begin{equation}\label{eq:M_single}
\begin{aligned}
 M&>
 \begin{cases}
 0,
 &~if~
 \dfrac{W}{\sqrt{2}\varepsilon}
 \le \dfrac{z_n(0)}{\overline b},
 \\[10pt]
 \dfrac{
 \dfrac{W}{\sqrt{2}\varepsilon}
 -\dfrac{z_n(0)}{\overline b}
 }{
 \displaystyle\int_0^{\tau_M}D(\tau)d\tau
 },
 &~if~
 \dfrac{W}{\sqrt{2}\varepsilon}
 >\dfrac{z_n(0)}{\overline b},
 \end{cases}
\end{aligned}
\end{equation}
where $\tau_M>0$ is the solution of
\begin{equation}\label{eq:tau_M}
 S(\tau_M)=
 \frac{\varepsilon z_n(0)}{\sqrt{2}\overline bW}
\end{equation}
when $ \dfrac{W}{\sqrt{2}\varepsilon}
 >\dfrac{z_n(0)}{\overline b}$, with $S$ defined in \eqref{eq:S_def}.
\begin{lemma}\label{lemma:tau_M}
The equation \eqref{eq:tau_M} has a unique bounded
solution $\tau_M>0$ when $\frac{W}{\sqrt{2}\varepsilon}
 >\frac{z_n(0)}{\overline b}$.
\end{lemma}
\begin{proof}
Considering $z_n(0)>0$ by Lemma~\ref{lemma:positive_initial}, when$ \dfrac{W}{\sqrt{2}\varepsilon}
 >\dfrac{z_n(0)}{\overline b}$, the right-hand side of \eqref{eq:tau_M}
lies in $(0,1/2)$. By Lemma~\ref{lemma:S}, $S(t)$ increases strictly from $S(0)=0$
to $\lim_{t\to\infty}S(t)=1/2$, so the intermediate value theorem gives a unique bounded $\tau_M>0$ solution to \eqref{eq:tau_M}.
\end{proof}

\subsection{Results of Safe Adaptive Control}
\begin{theorem}\label{theorem:adaptive}
For all initial data $u[0]\in L^2(0,1)$ and $Y(0)\in\mathbb R^n$ satisfying
Assumption~\ref{asmp:initial}, with the design parameters $\kappa_i$,
$i=1,\ldots,n-1$, satisfying \eqref{eq:kappa_i}, $\kappa_n$ satisfying
\eqref{eq:kappa_n_c}, and $M$ satisfying \eqref{eq:M1}, \eqref{eq:M_single}, the closed-loop
system including the plant \eqref{eq:plant1}--\eqref{eq:plant4} with the safe
adaptive controller \eqref{eq:Ua} has the following properties:
\begin{enumerate}[label=\arabic*)]
 \item All system states converge exponentially to zero in the sense of $\|u_x(\cdot,t)\|^2+\|u(\cdot,t)\|^2+|Y(t)|^2$. 
 \item The output $y_1(t)\geq0$ for all $t\geq0$.
 \item The control input $U(t)$ converges exponentially to zero.
\end{enumerate}
\end{theorem}
\begin{proof}
(1) On the time interval $[0,t_1)$, the gain $\check\theta(t)$ stays in the
compact set $\Theta$, so every kernel in \eqref{eq:Ua} is bounded and the
closed-loop system is linear with bounded coefficients; hence all its states
are bounded. For $t\ge t_1$, according to Lemma~\ref{lemma:identification}, we
directly have that the controller coincides with the nominal one. Considering
the Lyapunov function \eqref{eq:V} for $t\in[0,t_1]$ and $t\in[t_1,\infty]$,
respectively, choosing the analysis parameters $a_1$ and $a_2$ as
\begin{equation}\label{eq:a_choice_a}
 a_1>\max\left\{\frac{3|P\overline B|^2}{cq}-\frac{\varepsilon}{c},0\right\},
 \qquad a_2>\frac{3|P\overline B|^2}{cq},
\end{equation}
where the upper bound $\overline B=[0,\ldots,0,\overline b]^\top$ is used to
replace the unknown $B$, following the proof of property 1 in
Theorem~\ref{theorem:nominal}, together with the fact that $V(t)$ is continuous
at $t=t_1$, we thus obtain property 1 in this theorem.

(2) Similarly, like the proof of property 2 in Theorem~\ref{theorem:nominal},
considering \eqref{eq:bracket_lower} and \eqref{eq:z_n_bracket}, it suffices to show that
$z_n(0)+bM\int_0^tD(\tau)d\tau+\mathcal I(t)$ is positive. Since
$\|\cos(\omega_j\cdot)\|=1/\sqrt2$, we have 
$|\tilde\theta_j|\le\|w(\cdot,0)\|/\sqrt2$ by Cauchy--Schwarz inequality. Considering 
$\int_0^t e^{-\varepsilon\omega_j^2\tau}d\tau
=(1-e^{-\varepsilon\omega_j^2t})/(\varepsilon\omega_j^2)$ and \eqref{eq:S_def},
\begin{equation}\label{eq:I_bound_theta}
 \frac{|\mathcal I(t)|}{b}
 \le\frac{\sqrt2\|w(\cdot,0)\|}{\varepsilon}S(t)
 \le\frac{\sqrt2W}{\varepsilon}S(t),
\end{equation}
where the last inequality holds because $\theta\in\Theta$ and
$W=\max_{\ell\in\Theta}\|w_\ell(\cdot,0)\|$ given in \eqref{eq:W_S_definition}. Since
$b\le\overline b$, \eqref{eq:I_bound_theta} implies
\begin{equation}\label{eq:M_lower_estimate}
\begin{split}
 &\frac{1}{b}\Big[z_n(0)+bM\int_0^tD(\tau)d\tau+\mathcal I(t)\Big]
 \ge \frac{z_n(0)}{\overline b} \\
 &+M\int_0^tD(\tau)d\tau -\frac{\sqrt2W}{\varepsilon}S(t).
\end{split}
\end{equation}
When $W/(\sqrt2\varepsilon)\le z_n(0)/\overline b$,  the right-hand side of \eqref{eq:M_lower_estimate} is
positive for $M>0$ in \eqref{eq:M_single}, considering $S(t)<1/2$ by
Lemma~\ref{lemma:S}.
When $W/(\sqrt2\varepsilon)> z_n(0)/\overline b$,  for $0<t\le\tau_M$, the monotonicity of
$S$ and \eqref{eq:tau_M} give
$\frac{\sqrt2W}{\varepsilon}S(t)\le\frac{z_n(0)}{\overline b}$, so the
right-hand side is positive because $M>0$ and $D(\tau)>0$ on $(0,\infty)$; for
$t\ge\tau_M$, considering $\int_0^tD(\tau)d\tau\ge\int_0^{\tau_M}D(\tau)d\tau$ and $S(t)<1/2$, the second
branch of \eqref{eq:M_single} leads to
$M\int_0^tD(\tau)d\tau>\frac{W}{\sqrt2\varepsilon}-\frac{z_n(0)}{\overline b}
\ge\frac{\sqrt2W}{\varepsilon}S(t)-\frac{z_n(0)}{\overline b}$, hence the
right-hand side is also positive. Therefore $z_n(0)+b\int_0^t e^{c\tau}w(0,\tau)$ in
\eqref{eq:z_n_bracket} is positive, i.e., $z_n(t)>0$ for all $t\geq0$. The remaining
$z_i(t),i=n-1,...1$ equations then are
\begin{equation}\label{eq:z_i_a}
 z_i(t)=e^{-\kappa_i t}z_i(0)+\int_0^t
 e^{-\kappa_i(t-\tau)}z_{i+1}(\tau)d\tau>0,
\end{equation}
Since $y_1=z_1$, this proves $y_1(t)\geq0$, which is property 2.

(3) Finally, for property 3, for $t\ge t_1$ the estimates are exact and the
backstepping kernels in \eqref{eq:Ua} are the nominal kernels. Their
boundedness, the exponential convergence established in property 1, and
$\delta(t)=Me^{-ct}$ imply $U_a(t)\to0$ exponentially. The proof of this
theorem is complete.
\end{proof}

\section{Simulation}\label{sec:simulation}
The considered simulation model is \eqref{eq:plant1}--\eqref{eq:plant4} with
$A=\begin{bmatrix}0&1\\2&-1\end{bmatrix}$,
$B=[0,4]^\top$, $\varepsilon=1$, and $\lambda=10$, which is open-loop
unstable. 
The known bounds of the unknown parameters $\lambda$ and $b$ are set as
$\underline\lambda=8$, $\overline\lambda=12$, $\underline b=3$ and
$\overline b=7$, and the batch
least-squares identifier is initialized at $(\hat\lambda(0),\hat b(0))=(9,5)$.
Based on the simulation model given above, we construct the safe adaptive
controller $U_a(t)$ by choosing the design parameters as
$\kappa_1=6$, $\kappa_2=c=10$ and $M=246$. The arbitrary positive design
parameters are set as $T=0.5$, $\widetilde N=1$ and $\bar n=1$, and the initial
values are defined as $Y(0)=[0.2,-0.95]^\top$ and $u(x,0)=x^2\sin(\pi x)$. With
these data
the non-overshooting states are initialized as $z_1(0)=y_1(0)=0.2>0$ and
$z_2(0)=y_2(0)+\kappa_1z_1(0)=0.25>0$, so the nominal safe condition holds. The number of modes $\bar n$ is determined by the rank of the matrix $G_{\bar n}$: the two parameters are uniquely determined once $\operatorname{rank}G_{\bar n}=2$, so we start with a single mode ($\bar n=1$) and enlarge $\bar n$ only if the regressor is rank-deficient. In this simulation $\bar n=1$ already yields $\operatorname{rank}G_{\bar n}=2$.
For this second-order case, \eqref{eq:KT} gives
$K^\top=\frac{1}{\hat b(t)}[I_1+\kappa_1\kappa_2,\ I_2+\kappa_1+\kappa_2]$,
where $\hat b(t)$ is the second component of $\check\theta(t)$.
The simulation is conducted based on the finite difference method with a time
step of 0.001 and a space step of 0.05. For comparison, an adaptive controller is obtained
by keeping the same batch least-squares identifier but removing the safety mechanism, i.e., setting $\delta(t)=0$, and replacing the non-overshooting gain \eqref{eq:KT}
by a gain $K_{\rm ad}=[2,1]$ that is obtained from ensuring $A-BK_{\rm ad}$ Hurwitz.

\begin{figure}[t]
 \centering
 \begin{subfigure}[b]{0.49\columnwidth}
  \centering
  \includegraphics[width=\linewidth]{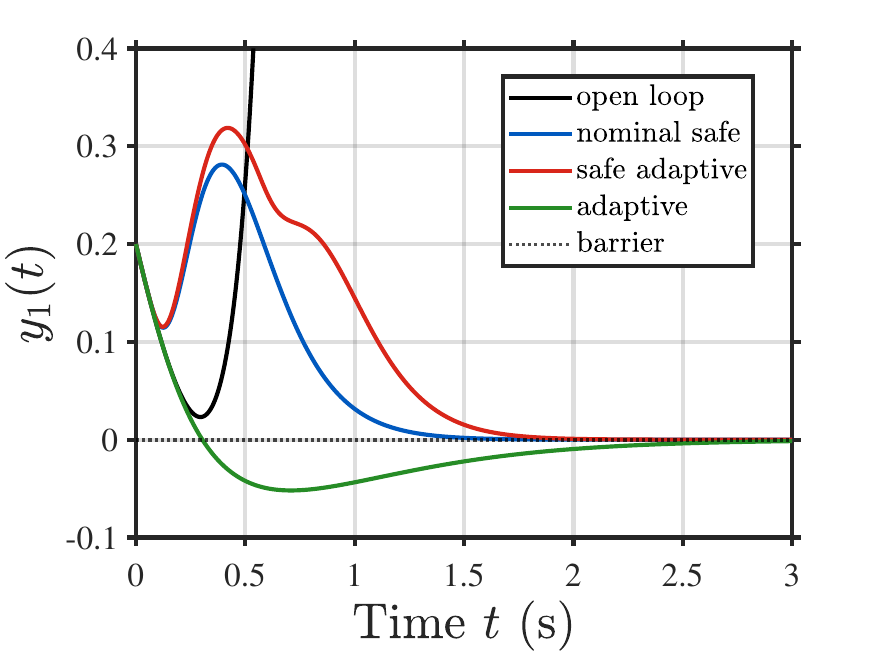}
  \caption{Output state $y_1(t)$.}
  \label{fig:sim_y1}
 \end{subfigure}\hfill
 \begin{subfigure}[b]{0.49\columnwidth}
  \centering
  \includegraphics[width=\linewidth]{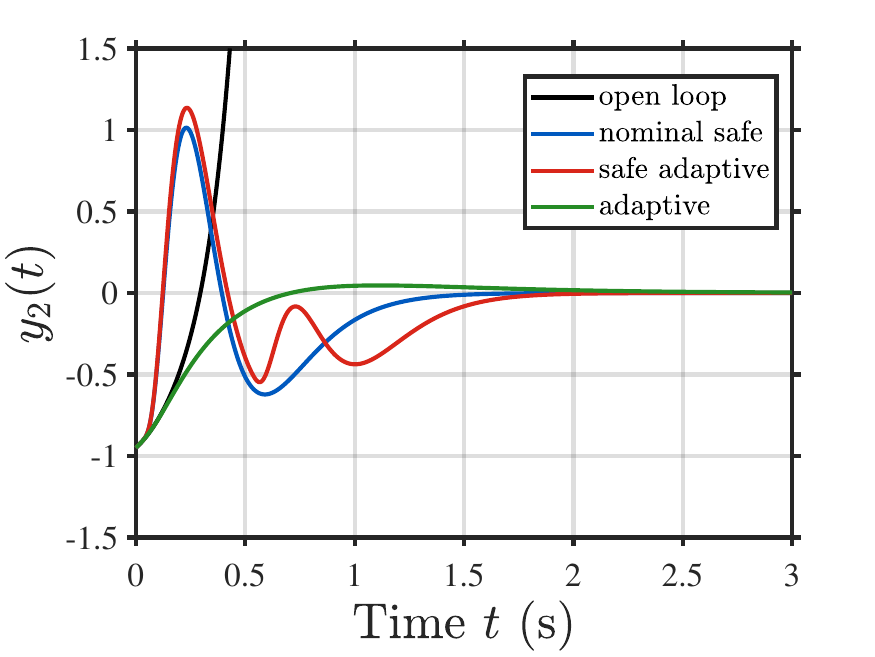}
  \caption{ODE state $y_2(t)$.}
  \label{fig:sim_y2}
 \end{subfigure}
 \caption{Open loop, nominal safe, safe adaptive, and adaptive controllers:
 (a) $y_1(t)$; (b) $y_2(t)$.}
 \label{fig:sim_pair}
\end{figure}

\begin{figure}[t]
 \centering
 \begin{subfigure}[b]{0.49\columnwidth}
  \centering
  \includegraphics[width=\linewidth]{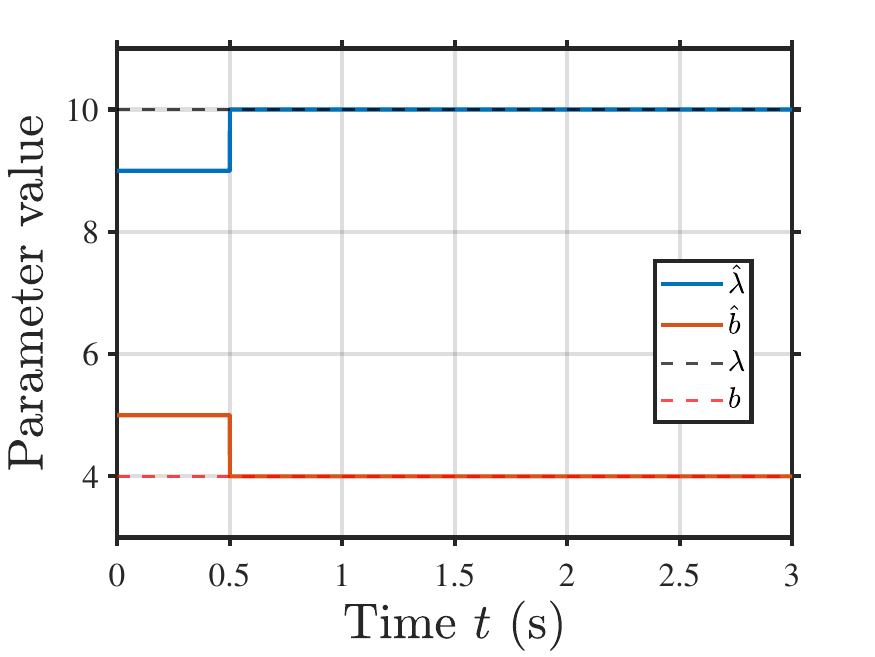}
  \caption{Parameter estimates.}
 \end{subfigure}\hfill
 \begin{subfigure}[b]{0.49\columnwidth}
  \centering
  \includegraphics[width=\linewidth]{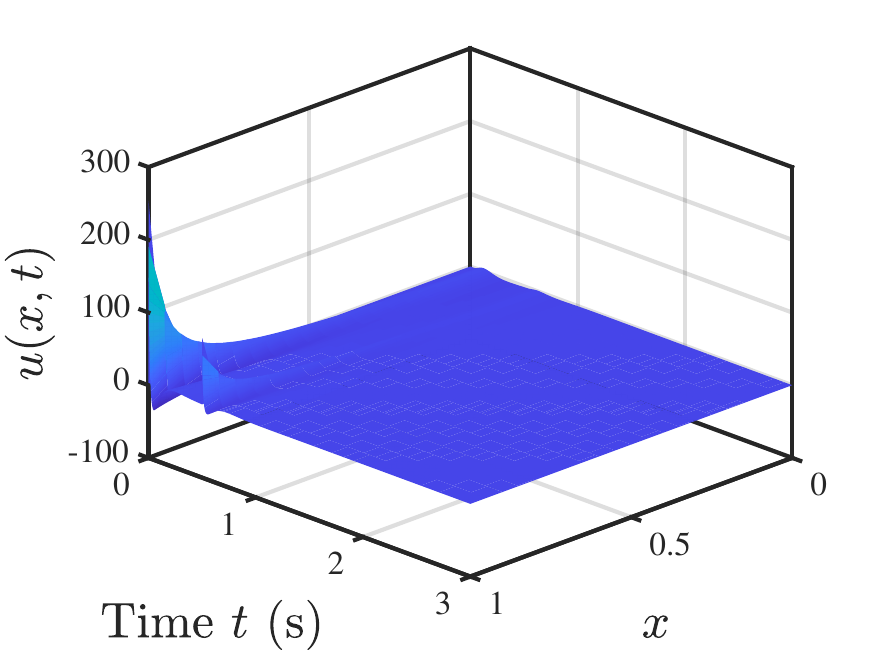}
  \caption{PDE state.}
 \end{subfigure}
 \caption{Parameter estimates and PDE state under the safe adaptive controller.}
 \label{fig:sim_aux}
\end{figure}

The results for $y_1(t)$ and $y_2(t)$ are shown in
Fig.~\ref{fig:sim_pair}(a) and Fig.~\ref{fig:sim_pair}(b), respectively, where
the black solid line denotes the open-loop response, the blue solid line denotes
the nominal safe controller, the red solid line denotes the safe adaptive
controller, and the green solid line denotes the adaptive controller. The
open-loop outputs diverge, since the simulation model is open-loop unstable. It can be seen that both the safe adaptive
controller and the nominal safe controller not only ensure that the output
converges to zero, but also keep $y_1(t)\geq0$ throughout the entire process,
whereas $y_1(t)$ under the adaptive controller violates the safety constraint,
reaching approximately $-6\times10^{-2}$. It is shown in Fig. \ref{fig:sim_pair}(b) that all controllers ensures convergence to zero of the ODE state $y_2$. The parameter estimates converge to
the true values at the first triggering time, as shown in
Fig.~\ref{fig:sim_aux}(a). The PDE state $u(x,t)$ under the safe adaptive
controller is shown in Fig.~\ref{fig:sim_aux}(b), which converges to zero.

\section{Conclusion and Future Work}
This paper presents a safe adaptive control strategy for a class of parabolic
PDE-ODE cascade systems subject to parameter uncertainties in both PDE and
ODE subsystems. The ODE non-overshooting transformation and the PDE
backstepping transformation are first combined for the known-parameter case,
and the resulting nominal controller preserves the output safety while
exponentially regulating the coupled states. The batch least-squares extension
uses the PDE modal equations and the last ODE equation to construct exact
consistency relations, and the parameter estimate becomes exact at the first
update instant. Under the proposed safe adaptive controller, the system output
is guaranteed to remain safe at all times if the initial state is safe.
Furthermore, all system states are ensured to converge to zero. The
effectiveness of the proposed control scheme is validated through numerical simulations. In our future work, a practical physical application will be considered,
incorporating the observer designs for system states and external
disturbances.
\section*{Appendix}
\appendix
\section{Solution to \eqref{eq:kernel_k}--\eqref{eq:kernel_rbc}}
\label{appendix:kernels}
According to \cite{li2020adaptive}, the solution to
\eqref{eq:kernel_k}--\eqref{eq:kernel_rbc} is given by
 $k(x,y)=\sum_{q=0}^{\infty}\bigl(G^{q+1}(\xi,\eta)-G^q(\xi,\eta)\bigr),$
where $\xi=x+y$, $\eta=x-y$, and
$G(\xi,\eta)=k((\xi+\eta)/2,(\xi-\eta)/2)$ satisfies
$
 G(\xi,\eta)=-\alpha(\xi+\eta)+\alpha\int_\eta^\xi\int_0^\eta
 G(m,\tau)d\tau dm+2\alpha\int_0^\eta\int_0^mG(m,\tau)d\tau dm
 +\int_0^\eta r(m)\frac{B}{\varepsilon}dm,
$
with $\alpha=(\lambda+c)/(4\varepsilon)$, $G^0=0$, and $G^{q+1}$ defined by the
right-hand side of the above integral equation with $G$ replaced by $G^q$. Moreover,
the solution of \eqref{eq:kernel_r}--\eqref{eq:kernel_rbc} is given by
\begin{equation}\label{eq:r_solution}
 r(x)=\begin{pmatrix}-K^\top&0\end{pmatrix}e^{Dx}\begin{pmatrix}I\\0\end{pmatrix},\qquad
 D=\begin{pmatrix}0&(A+cI)/\varepsilon\\I&0\end{pmatrix},
\end{equation}
where $I$ is the identity matrix.

\section{Solving $w(x,t)$ in \eqref{eq:target_pde}--\eqref{eq:target_dirichlet}}
\label{appendix:w}
The detailed process for solving $w(x,t)$ is shown as follows. Considering
\eqref{eq:v_transform} and \eqref{eq:target_pde}--\eqref{eq:target_dirichlet}, we have
\begin{equation}
 v_t=\varepsilon v_{xx}-cv-c\delta-\dot\delta,\label{eq:v_pde_src}
\end{equation}
and $v_x(0,t)=0, v(1,t)=0$. Noticing that $c\delta+\dot\delta$ in \eqref{eq:v_pde_src} is in fact zero,
because $\delta(t)=Me^{-ct}$. We nevertheless keep it in the derivation and
solve $w(x,t)$ in this general form to avoid repetitiveness, because this will
be reused when solving $w(x,t)$ in the subsequent safe adaptive controller
design. By eigenfunction expansion, we can express $v(x,t)$ and
$c\delta(t)+\dot\delta(t)$ as $v(x,t)=\sum_{j=0}^{\infty}v_j(t)\cos(\omega_jx)$ and
$c\delta(t)+\dot\delta(t)=\sum_{j=0}^{\infty}f_j(t)\cos(\omega_jx)$, where
$f_j(t)=\frac{2(-1)^j}{\omega_j}(c\delta+\dot\delta)$. Substitution into
\eqref{eq:v_pde_src} gives, for every $j\ge0$, $\dot v_j=-(\varepsilon\omega_j^2+c)v_j
 -\frac{2(-1)^j}{\omega_j}(c\delta+\dot\delta)$,
whose solution is
$ v_j(t) = e^{-(\varepsilon\omega_j^2+c)t}\theta_j 
-\frac{2(-1)^j}{\omega_j}\int_0^t
 e^{-(\varepsilon\omega_j^2+c)(t-\tau)}
 \bigl(c\delta(\tau)+\dot\delta(\tau)\bigr)d\tau,$ where $\theta_j$ is to be determined later. Since
$v(x,0)=w(x,0)-\delta(0)$, it is obtained that $\theta_j$ here are the Fourier
coefficients of $w(x,0)-\delta(0)$, that is,
$\theta_j=2\int_0^1[w(x,0)-\delta(0)]\cos(\omega_jx)dx$. Finally, we obtain
the solution of $w(x,t)$ as \eqref{eq:w_solution}.

\section{Proof of Lemma~\ref{lemma:D}}
\label{appendix:D}
Since $\omega_j=\frac{\pi}{2}(2j+1)$, put $x=\varepsilon(\frac{\pi}{2})^2t$ and
$L(x)=\sum_{j=0}^{\infty}\frac{(-1)^j}{2j+1}e^{-(2j+1)^2x}$, so that
$D(t)=1-\frac{4}{\pi}L(x)$. Let $F(x)=\frac{\pi}{4}-L(x)$. Since
$\sum_{j=0}^{\infty}\frac{(-1)^j}{2j+1}=\frac{\pi}{4}$, we have $F(0)=0$, and the
Jacobi theta identity
\begin{equation}
 \sum_{j=0}^{\infty}(-1)^j(2j+1)e^{-(2j+1)^2x}
 =e^{-x}\prod_{j=1}^{\infty}(1-e^{-8jx})^3
\end{equation}
\cite{whittaker2020modern} yields
$F'(x)=e^{-x}\prod_{j=1}^{\infty}(1-e^{-8jx})^3>0$ for $x>0$. Hence $F$ is
strictly increasing with $F(0)=0$ and $\lim_{x\to\infty}F(x)=\frac{\pi}{4}$, so
$0<L(x)<\frac{\pi}{4}$ and therefore $0<D(t)<1$ for all $t>0$. The proof is
complete.
\section{Proof of the convergence of the series}
\label{appendix:series}
We prove that the series in \eqref{eq:D}, \eqref{eq:w0_nonneg}, \eqref{eq:I_term}
and \eqref{eq:E_term} converge for all $t>0$. Proof. Since
$\omega_j^2=\pi^2(j+\tfrac12)^2$, the factor $e^{-\varepsilon\omega_j^2t}$
decays super-exponentially in $j$, so $\sum_{j}e^{-\varepsilon\omega_j^2t}$ is a
convergent series for all $t>0$. Considering \eqref{eq:w0_nonneg}, and applying
Cauchy--Schwarz Inequality, we have
$
 \sum_{j=0}^{\infty}e^{-\varepsilon\omega_j^2t}|\tilde\theta_j|
 \le \sum_{j=0}^{\infty}e^{-\varepsilon\omega_j^2t}
 \sqrt{\frac{1}{2}\int_0^1w(x,0)^2dx}
 \le \sqrt{\frac{1}{2}\int_0^1w(x,0)^2dx}
 \sum_{j=0}^{\infty}e^{-\varepsilon\omega_j^2t},
$
for any $t>0$. Since $\sum_{j}e^{-\varepsilon\omega_j^2t}$ is a convergent
series for all $t>0$, we can conclude that the series in \eqref{eq:w0_nonneg}
and \eqref{eq:I_term} is absolutely convergent for all $t>0$. Next, we consider
the alternating series in \eqref{eq:D}. For any $t>0$, we have that
$e^{-\varepsilon\omega_j^2t}/\omega_j$ is monotonically decreasing as $j$
increases and approaches zero as $j\to+\infty$. Therefore, according to Leibniz
criterion for alternating series, we can conclude that the series in
\eqref{eq:D} is absolutely convergent for all $t>0$. Finally, an integration by
parts in \eqref{eq:E_term} gives the convergence of $\mathcal E(t)$. The proof
is complete.

\small
\bibliographystyle{plain}
\bibliography{reference}
\end{document}